\documentclass[a4paper, journal, twocolumn]{IEEEtran}
\IEEEoverridecommandlockouts

\usepackage{bm}
\usepackage{cite}
\usepackage{amsmath, amssymb, amsthm}
\usepackage{graphicx, subfigure}
\usepackage[noend]{algpseudocode}
\usepackage{algorithmicx,algorithm}
\usepackage{color}

\providecommand{\mathbold}[1]{\bm{#1}}
\newcommand{\vct}[1]{\mathbold{#1}}
\newcommand{\mtx}[1]{\mathbold{#1}}
\def \R 	{\mathbb{R}}
\def \P 	{\mathbb{P}}
\def \E 	{\mathbb{E}}

\newcommand{\mPsi}{\mtx{\Psi}}
\newcommand{\mA}{\mtx{A}}
\newcommand{\mC}{\mtx{C}}
\newcommand{\mD}{\mtx{D}}

\newcommand{\mU}{\mtx{U}}
\newcommand{\mI}{\mtx{I}}
\newcommand{\vu}{\vct{u}}
\newcommand{\va}{\vct{a}}
\newcommand{\vb}{\vct{b}}
\newcommand{\vx}{\vct{x}}
\newcommand{\vy}{\vct{y}}
\newcommand{\vz}{\vct{z}}
\newcommand{\vr}{\vct{r}}
\newcommand{\vs}{\vct{s}}
\newcommand{\vt}{\vct{t}}
\newcommand{\vv}{\vct{v}}

\newcommand{\vg}{\vct{g}}
\newcommand{\vh}{\vct{h}}

\newcommand{\Ds}{\mathcal{D}_s}
\newcommand{\Dc}{\mathcal{D}_c}

\newcommand{\argmin}{\operatorname*{arg\; min}}
\newcommand{\argmax}{\operatorname*{arg\; max}}

\newtheorem{theorem}{Theorem} 
\newtheorem{lemma}{Lemma}

\newtheorem{remark}{Remark}

\begin{document}

\title{On the Phase Transition of Corrupted Sensing}

\author{
  \IEEEauthorblockN{Huan Zhang\IEEEauthorrefmark{1}\IEEEauthorrefmark{2}, Yulong Liu\IEEEauthorrefmark{3}, and Hong Lei\IEEEauthorrefmark{1}}
	
\IEEEauthorblockA{\IEEEauthorrefmark{1}Institute of Electronics, Chinese Academy of Sciences, Beijing 100190, China}
	
\IEEEauthorblockA{\IEEEauthorrefmark{2}University of Chinese Academy of Sciences, Beijing 100049, China}
	
\IEEEauthorblockA{\IEEEauthorrefmark{3}School of Physics, Beijing Institute of Technology, Beijing 100081, China}

\thanks{This work was supported by the National Natural Science Foundation of China under Grant 61301188.}

}

%

\maketitle

\pagestyle{empty}  
\thispagestyle{empty} 

\begin{abstract}
  In \cite{FOY2014}, a sharp phase transition has been numerically observed when a constrained convex procedure is used to solve the corrupted sensing problem. In this paper, we present a theoretical analysis for this phenomenon. Specifically, we establish the threshold below which this convex procedure fails to recover signal and corruption with high probability. Together with the work in \cite{FOY2014}, we prove that a sharp phase transition occurs around the sum of the squares of spherical Gaussian widths of two tangent cones. Numerical experiments are provided to demonstrate the correctness and sharpness of our results.
\end{abstract}


%
\IEEEpeerreviewmaketitle

\section{Introduction} \label{sec:introduction}
Corrupted sensing aims to recover a structured signal from a small number of corrupted measurements
\begin{equation} \label{corrupted sensing model}
	\vy = \mPsi \vx^{\star} + \vv^{\star},
\end{equation}
where $\mPsi \in \R^{m \times n}$ is the sensing measurement matrix which is assumed to have i.i.d. standard Gaussian entries in this paper, $\vx^{\star} \in \R^n$ is the unknown signal, and $\vv^{\star} \in \R^m$ is an unknown corruption. The goal is to estimate $\vx^{\star}$ and $\vv^{\star}$ from $\vy$ and $\mPsi$.

This problem is encountered in many practical applications, such as face recognition \cite{WRIGHT2009}, subspace clustering\cite{ELHAM2009}, network data analysis \cite{HAUPT2008}, and so on. Theoretical guarantees for this problem include sparse signal recovery from sparse corruption \cite{ WRIGHT2010, LI2013, NGUYEN2013, NGUYEN2013_2, POPE2013, STUDER2012, STUDER2014} and structured signal recovery from structured corruption \cite{FOY2014, mccoy2014sharp, Chen2017}.

To make the recovery possible, we will assume that both $\vx$ and $\vv$ have some structures which are promoted by the convex functions $f(\cdot)$ and $g(\cdot)$ respectively. When prior information about $f(\vx^{\star})$ or $g(\vv^{\star})$ is available, it is natural to consider the following program to recover the signal and corruption:
	\begin{equation} \label{constrained convex program 1}
	  	\min f(\vx), \quad \text{s.t. } \vy = \mPsi \vx + \vv, \quad g(\vv) \le g(\vv^{\star}),
  	\end{equation}
	or
	\begin{equation} \label{constrained convex program 2}
	  	\min g(\vv), \quad \text{s.t. } \vy = \mPsi \vx + \vv, \quad f(\vx) \le f(\vx^{\star}).
	\end{equation}

In \cite{FOY2014}, Foygel and Mackey provided conditions under which convex program (\ref{constrained convex program 1}) or (\ref{constrained convex program 2}) succeeds with high probability. Numerical experiments in \cite{FOY2014} also suggested that there is a sharp phase transition when (\ref{constrained convex program 1}) or (\ref{constrained convex program 2}) is used to solve the corrupted sensing problem.
However, little work has devoted to determining the threshold below which (\ref{constrained convex program 1}) or (\ref{constrained convex program 2}) fails with high probability. Therefore, theoretical understanding of the phase transition for program (\ref{constrained convex program 1}) and (\ref{constrained convex program 2}) is far from satisfactory.


In this paper, we present a theoretical analysis for the phase transition of (\ref{constrained convex program 1}) or (\ref{constrained convex program 2}).
In particular, we figure out the exact position of phase transition, and demonstrate that the phase transition occurs in a relatively narrow region.

\section{Preliminaries} \label{sec:preliminaries}

In this section, we present some preliminaries which will be used in our analysis.

Our result involves two important concepts: the Gaussian width and the tangent cone. Given a subset $T$ in $\R^n$, the Gaussian width is defined by
$$
\omega(T) = \E \sup_{\vt \in T} \left< \vg, \vt \right>, \text{ where }\vg \sim N(0,I_n).
$$
We also define two tangent cones corresponding to signal and corruption respectively. The tangent cone of $f(\cdot)$ at the true signal $\vx^{\star}$ is defined as
\begin{equation} \label{tangent cone for signal}
  	\mathcal{D}_s = \big\{ \vct{a} \in \R^n: \exists ~t > 0, f(\vx^{\star} + \vct{a}t) \le f(\vx^{\star}) \big\}.
\end{equation}
Similarly, the tangent cone of $g(\cdot)$ at the true corruption $\vv^{\star}$ is given by
\begin{equation} \label{tangent cone for corruption}
  	\mathcal{D}_c = \big\{ \vct{b} \in \R^m: \exists ~t > 0, g(\vv^{\star} + \vct{b}t) \le g(\vv^{\star}) \big\}.
\end{equation}

\section{Main results} \label{sec:main results}
In this section, we state our main results with some discussions.
\begin{theorem} [Failure of convex program (\ref{constrained convex program 1}) or (\ref{constrained convex program 2})] \label{th:main theorem 1}
  Consider convex program (\ref{constrained convex program 1}) or (\ref{constrained convex program 2}). Assume that both tangent cones $\Ds$ and $\Dc$ are closed. For any $t \geq 0$, if the measurement number $m$ satisfies
	$$
	\sqrt{m} < \sqrt{\omega^2 \big(\mathcal{D}_s \cap S^{n-1} \big) + \omega^2 \big(\mathcal{D}_c \cap S^{m-1} \big)} - t,
	$$
	then the constrained convex program (\ref{constrained convex program 1}) or (\ref{constrained convex program 2}) fails with probability at least $1-\exp(-t^2/2)$, where $S^{n-1}$ and $S^{m-1}$ are the unit sphere of $\R^n$ and $\R^m$ respectively.
\end{theorem}

\begin{proof}
  	See Appendix \ref{app:A}.
\end{proof}


\begin{remark} [Phase transition of corrupted sensing] \label{remark by combining mackey}
  	Recall Theorem $1$ and Remark $2$ in \cite{FOY2014}, which stated that \footnote{The authors believe that the small additive constants are artifacts of the proof technique.} \footnote{The original result is stated in terms of Gaussian complexity $\gamma(\Ds \cap B^n)$, difined as $\gamma^2(\Ds \cap B^n) = \E \big( \sup_{t \in \Ds \cap B^n} \left< g, t \right> \big)^2$, where $B^n$ denotes the $\ell_2$ unit ball in $\R^n$. However, as the author stated, the Gaussian complexity $\gamma(\Ds \cap B^n)$ is only very slightly larger than $\omega(\Ds \cap S^{n-1})$.} when
	$$
	\sqrt{m} \ge \sqrt{\omega^2(\Ds \cap S^{n-1}) + \omega^2(\Dc \cap S^{m-1})} + \frac{1}{\sqrt{2}} + \frac{1}{\sqrt{2 \pi}} + t,
	$$
	the constrained convex program (\ref{constrained convex program 1}) or (\ref{constrained convex program 2}) succeeds with probability at least $1-\exp(-t^2/2)$. This, together with our result Theorem \ref{th:main theorem 1}, demonstrate that the phase transition of corrupted sensing occurs around
	$$
	\omega^2 \big(\mathcal{D}_s \cap S^{n-1} \big) + \omega^2 \big(\mathcal{D}_c \cap S^{m-1} \big),
	$$
	and the width of phase transition area is about
	$$
	C\sqrt{\omega^2 \big(\mathcal{D}_s \cap S^{n-1} \big) + \omega^2 \big(\mathcal{D}_c \cap S^{m-1} \big)},
	$$
	where $C$ is an absolute constant.
\end{remark}

\begin{remark} \label{relation with statistical dimension}
  	Our result also agrees with the result of Amelunxen \textit{el al.} \cite{AMEL2014}. Indeed, by Proposition 10.2 and Proposition 3.1 (9) in \cite{AMEL2014}, we have
	$$
	\omega^2 \big(\mathcal{D}_s \cap S^{n-1} \big) + \omega^2 \big(\mathcal{D}_c \cap S^{m-1} \big) \approx \delta(\Ds) + \delta(\Dc) = \delta(\Ds \times \Dc),
	$$
	where $\delta(\mathcal{D})$ denotes the statistical dimension of a convex cone $\mathcal{D}$.
\end{remark}

\begin{remark} \label{compare with AMEL2014}
  	In \cite{AMEL2014}, Amelunxen \textit{et al.} considered the phase transition of the following demixing problem:
  	$$
  	\vz = \vx + \mU \vy,
  	$$
	where $\vx$, $\vy \in \R^n$ are unknown signals and $\mU \in \R^{n \times n}$ is a random orthogonal matrix. This model is different from ours since we have random Gaussian measurement matrix with $m \ll n$.
\end{remark}

\begin{remark} \label{compare with OYMAK2015}
  	In \cite{OYMAK2015}, Oymak and Tropp considered the phase transition of the following demixing model:
  	$$
  	\vy = \mPsi_0 \vx_0 + \mPsi_1 \vx_1,
  	$$
	where $\vx_0$, $\vx_1 \in \R^n$ are two signals and $\mPsi_0$, $\mPsi_1 \in \R^{m \times n}$ are some random transformation matrices. 
	This model is also different from ours since $\mPsi_1$ is a deterministic matrix in our case. This makes the problem more difficult to analyze. 
\end{remark}

\section{Simulation Results}
In this section, we employ a numerical experiment to verify our theoretical guarantees (Theorem \ref{th:main theorem 1}). In the experiment, both signal and corruption are designed to be sparse vectors. We use CVX \cite{CVX1} \cite{CVX2} to solve the convex program (\ref{constrained convex program 1}) or (\ref{constrained convex program 2}).

In the experiment, we assume that the prior information of $f(\vx^{\star})$ is known exactly, and solve program (\ref{constrained convex program 2}). The experiment settings are as follows: the ambient dimension $n$ is set to $128$, the measurement number $m=n=128$, the sparsity level of signal changes from $1$ to $n$ with step $1$, and the same for corruption. For every sparsity level of signal and corruption, we run and solve (\ref{constrained convex program 2}) $20$ times. We declare success if the solution to (\ref{constrained convex program 2}), denoted by $(\hat{\vx}, \hat{\vv})$, satisfies $\|\hat{\vx} - \vx^{\star}\|_2 \le 10^{-3}$. 
Then we get the empirical probability of successful recovery. At last, we plot the theoretical curve predicted by Theorem \ref{th:main theorem 1}.

Our numerical experiment result is shown in Fig. \ref{fig:phase transition 1}. We can see that the theoretical threshold given by Theorem \ref{th:main theorem 1} is closely matched with the empirical phase transition. It means that our theory can give a reliable prediction of the phase transition curve.
\begin{figure}
  	\centering
  	\includegraphics[width=.45\textwidth]{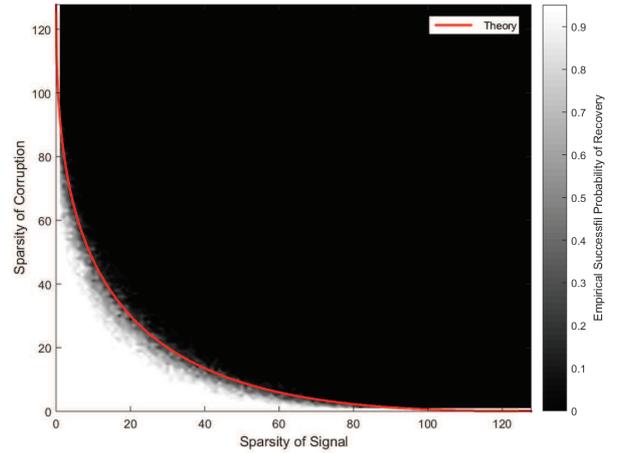}
	\caption{Phase transition for constrained convex program \ref{constrained convex program 2}. The red curve plots the phase transition threshold predicted by Theorem \ref{th:main theorem 1}.}
	\label{fig:phase transition 1}
\end{figure}

\section{Conclusion} \label{sec:conclusion}
This paper studied the problem of phase transition when we use convex program to solve corrupted sensing problem. Our results, together with previous work \cite{FOY2014}, gave the exact location of phase transition and the size of transition region. Simulations were provided to verify the correctness of our results. Our ongoing work is to establish a general framework to analyze the phase transition of various convex programs with noise-free or noisy data.

\appendices
\section{Proof of Main Results} \label{app:A}
In this section, we present proof for our main result (Theorem \ref{th:main theorem 1}). First, we will establish a sufficient condition under which convex program (\ref{constrained convex program 1}) or (\ref{constrained convex program 2}) fails, then some necessary tools are introduced, and at last, we give the proof for Theorem \ref{th:main theorem 1}.

\subsection{Sufficient Condition for failure}
In this subsection, we establish an easy-to-handle sufficient condition under which program (\ref{constrained convex program 1}) or (\ref{constrained convex program 2}) fails.
\begin{lemma} \label{sufficient condition for failure: step 1}
  	Let $\mathcal{D}_s$ and $\mathcal{D}_c$ denote the signal and the corruption tangent cones defined in (\ref{tangent cone for signal}) and (\ref{tangent cone for corruption}) respectively. Then a sufficient condition under which constrained convex program (\ref{constrained convex program 1}) or (\ref{constrained convex program 2}) fails is
	\begin{equation} \label{eq:sufficient condition for failure: step 1}
	  	\min_{(\va,\vb) \in (\Ds \times \Dc) \cap S^{n+m-1}} \big\| \mPsi \va + \vb \big\| = 0.
  	\end{equation}
	In other words, the subset $\Ds \times \Dc \cap S^{n+m-1}$ intersects the null space of matrix $
	\begin{bmatrix}
	  	\mPsi & \mI
	\end{bmatrix}$.
\end{lemma}

\begin{proof}
  	Lemma \ref{sufficient condition for failure: step 1} is a generalization of Proposition 2.1 of \cite{CHAN2012}. The proof is similar, and hence is omitted.
\end{proof}



Although Lemma \ref{sufficient condition for failure: step 1} gives a sufficient condition for failure, it is difficult to check when (\ref{eq:sufficient condition for failure: step 1}) holds. The following lemma can overcome this drawback.

\begin{lemma} [Sufficient condition for failure, Proposition 3.8, \cite{OYMAK2015}] \label{sufficient condition for failure: step 2}
  	Under the condition of Lemma \ref{sufficient condition for failure: step 1}, if both $\Ds$ and $\Dc$ are closed, a sufficient condition for (\ref{eq:sufficient condition for failure: step 1}) to hold is
	\begin{equation} \label{eq:sufficient condition for failure: step 2}
	  	\min_{\| \vr \| = 1} \min_{\vs \in (\Ds \times \Dc)^{\circ}} \big\| \vs - \mA^* \vr \big\| > 0,
	\end{equation}
	where $(\Ds \times \Dc)^{\circ}$ denotes the polar cone of $\Ds \times \Dc$, $\mA =
	\begin{bmatrix}
	  	\mPsi & \mI
	\end{bmatrix}$, and $\mI$ denotes the identity matrix.
\end{lemma}

\begin{remark} \label{remark for sufficient condition}
  	One can easily check that 
	$$
	(\Ds \times \Dc)^{\circ} = \Ds^{\circ} \times \Dc^{\circ}.
	$$
	Thus, the sufficient condition under which convex program (\ref{constrained convex program 1}) or (\ref{constrained convex program 2}) fails can be rewritten as
	\begin{equation} \label{eq:remark for sufficient condition}
	  	\min_{\| \vr \| = 1} \min_{\vs \in \Ds^{\circ} \times \Dc^{\circ}} \big\| \vs - \mA^*\vr \big\| > 0.
	\end{equation}
\end{remark}
In the following parts, we will prove that (\ref{eq:remark for sufficient condition}) holds with high probability when the condition of Theorem \ref{th:main theorem 1} is satisfied. Before this, let's state some tools that will be used in our proof.
  	
\subsection{Other Useful Tools}

\begin{lemma}[Gordon's inequality, Theorem 3.16, \cite{LEDOUX1991}] \label{Gordon's inequality}
  	Let $(X_{\vu \vt})_{\vu \in U, \vt \in T}$ and $(Y_{\vu \vt})_{\vu \in U, \vt \in T}$ be two Gaussian processes indexed by pairs of points $(\vu,\vt)$ in a product set $U \times T$. Assume that
	$$
	\E (X_{\vu \vt}-X_{\vu \vs})^2 \le \E (Y_{\vu \vt} - Y_{\vu \vs})^2 \quad \text{for all }\vu,\vt,\vs; \\
	$$
	$$
	\E (X_{\vu \vt}-X_{\vv \vs})^2 \ge \E (Y_{\vu \vt} - Y_{\vv \vs})^2 \quad \text{for all }\vu \neq \vv \text{ and all }\vt,\vs.
	$$
	Then we have
	$$
	\E \inf_{\vu \in U} \sup_{\vt \in T} X_{\vu \vt} \le \E \inf_{\vu \in U} \sup_{\vt \in T} Y_{\vu \vt}.
	$$
\end{lemma}

\begin{lemma}[Concentration of measure, Theorem 5.6, \cite{BOUCH2013}] \label{concentration of measure}
  	Let $X = (X_1, \dots, X_n)$ be a vector of $n$ independent standard normal random variables. Let $f$ : $\R^n \rightarrow \R$ denotes an L-Lipschitz function. Then, for all $t \geq 0$,
	$$
	\P \big\{ f(X) - \E f(X) \ge t \big\} \le e^{-t^2/(2L^2)}.
	$$
\end{lemma}

\begin{lemma} [Lemma 3.7, \cite{CHAN2012}] \label{relation between GW of cone and polar cone}
  	Let $\mathcal{D} \subset \R^n$ be a non-empty closed, convex cone. Then we have that
	$$
	\omega^2(\mathcal{D} \cap S^{n-1}) + \omega^2(\mathcal{D}^{\circ} \cap S^{n-1}) \le n.
	$$
\end{lemma}

\begin{lemma} \label{Lipschitz function}
  	Let $\Omega_1$ and $\Omega_2$ be subsets of $S^{m-1}$ and $S^{n-1}$ respectively. Then the function
	$$
	F(\mPsi) = \min_{\vt \in \Omega_1} \max_{\vu \in \Omega_2} \left< \mPsi \vu, \vt \right>
	$$
	is a 1-Lipschitz function, where $\mPsi$ is the same as in (\ref{corrupted sensing model}).
\end{lemma}

\begin{proof}
  	See Appendix \ref{app:B}.
\end{proof}

\subsection{Proof of Main Results}
According to Remark \ref{remark for sufficient condition}, we only need to prove that when
$$
\sqrt{m} < \sqrt{\omega^2 \big(\mathcal{D}_s \cap S^{n-1} \big) + \omega^2 \big(\mathcal{D}_c \cap S^{m-1} \big)} - t,
$$
the following event
$$
\min_{\| \vr \| = 1} \min_{ \vs \in \Ds^{\circ} \times \Dc^{\circ}} \big\| \vs - \mA^*\vr \big\| > 0
$$
holds with probability at least $1-e^{-t^2/2}$. Moreover, a simple calculation verifies that this inequality is equivalent to
\begin{align} \label{equivalent condition for sufficient condition}
  	& \min_{\|\vr\|=1} \min_{\vs \in \Ds^{\circ} \times \Dc^{\circ}} \|\vs - \mA^* \vr \|_2 > 0 \quad \notag\\
	& \hspace{2pt} \Longleftrightarrow \quad \min_{\|\vr\|=1} \min_{\vs \in \Ds^{\circ} \times \Dc^{\circ}} \|\vs - \mA^* \vr \|_2^2 > 0 \notag\\
	& \hspace{2pt} \Longleftrightarrow \quad \min_{\|\vr\|=1} \min_{\vs_1 \in \Ds^{\circ} \atop \vs_2 \in \Dc^{\circ}} \Big[ \big\| \vs_1 - \mPsi^* \vr \big\|_2^2 + \big\| \vs_2 - \vr \big\|_2^2 \Big] > 0.
\end{align}
Now, we will consider two cases for $\vr$: \\
\textbf{Case I: $\vr \in \Dc^{\circ} \cap S^{m-1}$}. In this case, when we minimize over $\vs_2$, the second term $\big\| \vs_2 - \vr \big\|_2^2$ will be zero. Thus, the above inequality (\ref{equivalent condition for sufficient condition}) is equivalent to
\begin{align} \label{equivalence 2}
  	& \min_{\vr \in \Dc^{\circ} \cap S^{m-1}} \min_{\vs_1 \in \Ds^{\circ} \atop \vs_2 \in \Dc^{\circ}} \Big[ \big\| \vs_1 - \mPsi^* \vr \big\|_2^2 + \big\| \vs_2 - \vr \big\|_2^2 \Big] > 0 \notag\\
	& \hspace{35pt}\Longleftrightarrow \quad \min_{\vr \in \Dc^{\circ} \cap S^{m-1}} \min_{\vs_1 \in \Ds^{\circ}} \big\| \vs_1 - \mPsi^* \vr \big\|_2^2 > 0 \notag \\
	& \hspace{35pt}\Longleftrightarrow \quad \min_{\vr \in \Dc^{\circ} \cap S^{m-1}} \min_{\vs_1 \in \Ds^{\circ}} \big\| \vs_1 - \mPsi^* \vr \big\|_2 > 0.
\end{align}
For our purpose, we need to lower bound the left side of (\ref{equivalence 2}). 
Note that for any fixed $\vr \in \Dc^{\circ} \cap S^{m-1}$, we have
\begin{align*}
  	\min_{\vs_1 \in \mathcal{D}_s^{\circ}} \|\vs_1 - \mPsi^* \vr \|_2
  	&= \min_{\vs_1 \in \mathcal{D}_s^{\circ}} \max_{\vu \in S^{n-1}} \left<\vu, \mPsi^* \vr -\vs_1 \right> \\
  	&\ge \max_{\vu \in S^{n-1}} \min_{\vs_1 \in \mathcal{D}_s^{\circ}} \left<\vu, \mPsi^* \vr -\vs_1 \right> \\
  	&= \max_{\vu \in S^{n-1}} \big[ \left< \vu, \mPsi^* \vr \right> - \max_{\vs \in \mathcal{D}_s^{\circ}} \left<\vu, \vs \right> \big] \\
  	&= \max_{\vu \in \mathcal{D}_s \cap S^{n-1}} \left< \vu, \mPsi^* \vr \right> \\
  	&= \max_{\vu \in \mathcal{D}_s \cap S^{n-1}} \left< \mPsi \vu, \vr \right>.
\end{align*}
The first equality is due to the definition of $\ell_2$-norm. The first inequality is because of the minimax inequality. The second equality comes from the linear property of inner product. The third equality uses the fact that $\max_{\vs \in \mathcal{D}_s^{\circ}} \left<\vu, \vs \right> = 0$ when $\vu \in \mathcal{D}_s$, otherwise it equals $\infty$. The last equality can be derived by a simple transformation. As the above inequality holds for any $\vr \in \Dc^{\circ} \cap S^{m-1}$, we have
\begin{multline} \label{bound for min min norms}
	\min_{\vr \in \mathcal{D}_c^{\circ} \cap S^{m-1}} \min_{\vs_1 \in \mathcal{D}_s^{\circ}} \|\vs_1 - \mPsi^* \vr \|_2 \\
	\ge \min_{\vr \in \mathcal{D}_c^{\circ} \cap S^{m-1}} \max_{\vu \in \mathcal{D}_s \cap S^{n-1}} \left< \mPsi \vu, \vr \right>.
\end{multline}
It remains to bound the right side. To this end, we will first use Gordon's inequality (Lemma \ref{Gordon's inequality}) to derive a lower bound for the expectation, and then concentration of measure (Lemma \ref{concentration of measure}) to obtain the desired result. Let $X_{\vr \vu} := \left< \mPsi \vu, \vr \right>$ and $Y_{\vr \vu} := \left< \vg, \vr \right> + \left< \vh, \vu \right>$ be two Gaussian processes, where $\vg \sim N(\vct{0}, \mI_{m \times m})$ and $\vh \sim N(\vct{0}, \mI_{n \times n})$ are independent standard Gaussian random vectors. It can be easily checked that the increments satisfy
$$
\E (X_{\vr \vu} - X_{\vr \vu'})^2 = \big\| \vu - \vu' \big\|_2^2 = \E (Y_{\vr \vu} - Y_{\vr \vu'})^2,
$$
\begin{align*}
	\E (X_{\vr \vu} - X_{\vr' \vu'})^2
	&= \big\| \vu \vr^T - \vu' \vr'^T\big\|_F^2 \\
	&\le \big\| \vu - \vu' \big\|_2^2 + \big\| \vr - \vr'\big\|_2^2 \\
	&= \E (Y_{\vr \vu} - Y_{\vr' \vu'})^2.
\end{align*}
Therefore, Gordon's inequality (Lemma \ref{Gordon's inequality}) gives us:
\begin{align} \label{eq:comparison result}
  	& \E \min_{\vr \in \Dc^{\circ} \cap S^{m-1}} \max_{\vu \in \Ds \cap S^{n-1}} X_{\vr \vu} \notag\\
	& \hspace{25pt}\ge \E \min_{\vr \in \Dc^{\circ} \cap S^{m-1}} \max_{\vu \in \Ds \cap S^{n-1}} Y_{\vr \vu} \notag \\
	& \hspace{25pt} =  \E \min_{\vr \in \Dc^{\circ} \cap S^{m-1}} \left<\vg, \vr\right> + \E \max_{\vu \in \Ds \cap S^{n-1}} \left<\vh, \vu\right>.
\end{align}
Since $\vg$ is a symmetric random vector, we have
\begin{align*}
	\E \min_{\vr \in \Dc^{\circ} \cap S^{m-1}} \left<\vg, \vr\right> &= \E \min_{\vr \in \Dc^{\circ} \cap S^{m-1}} \left<-\vg, \vr\right> \\
	&= - \E \max_{\vr \in \Dc^{\circ} \cap S^{m-1}} \left<\vg, \vr\right> \\
	&= - \omega(\Dc^{\circ} \cap S^{m-1}).
\end{align*}
Substituting this into (\ref{eq:comparison result}), we get
\begin{equation} \label{eq:Gordon's result}
	\E \min_{\vr \in \Dc^{\circ} \cap S^{m-1}} \max_{\vu \in \Ds \cap S^{n-1}} X_{\vr \vu} \ge \omega(\Ds \cap S^{n-1}) - \omega(\Dc^{\circ} \cap S^{m-1}).
\end{equation}
As $\Dc$ is a closed convex cone,
by Lemma \ref{relation between GW of cone and polar cone}, we know that
$$
\hspace{13pt}\omega^2 \big(\mathcal{D}_c^{\circ} \cap S^{m-1} \big) + \omega^2 \big(\mathcal{D}_c \cap S^{m-1} \big)
\le m,
$$
which implies
$$
\omega(\mathcal{D}_c^{\circ} \cap S^{m-1}) \le \sqrt{m - \omega^2(\mathcal{D}_c \cap S^{m-1})}.
$$
Substituting this into (\ref{eq:Gordon's result}), we get the following result:
\begin{align} \label{eq:bound for expectation}
  	&\E \min_{\vr \in \mathcal{D}_c^{\circ} \cap S^{m-1}} \max_{\vu \in \mathcal{D}_s \cap S^{n-1}} \left< \mPsi \vu, \vr \right> \notag\\
	& \hspace{20pt}\ge  \omega(\mathcal{D}_s \cap S^{n-1}) - \sqrt{m - \omega^2(\mathcal{D}_c \cap S^{m-1})} \notag \\
	& \hspace{20pt}\ge \sqrt{\omega^2(\mathcal{D}_s \cap S^{n-1}) + \omega^2(\mathcal{D}_c \cap S^{m-1})} - \sqrt{m}.
\end{align}
In the last inequality, we have used the assumption that $\omega^2(\mathcal{D}_s \cap S^{n-1}) + \omega^2(\mathcal{D}_c \cap S^{m-1}) > m$.

Next, Lemma \ref{Lipschitz function} confirms that the following function
$$
\min_{\vr \in \mathcal{D}_c^{\circ} \cap S^{m-1}} \max_{\vu \in \mathcal{D}_s \cap S^{n-1}} \left< \mPsi \vu, \vr \right>
$$
is a $1$-Lipschitz function. Thus, concentration of measure (Lemma \ref{concentration of measure}) gives us that for any $t \geq 0$,
\begin{align*}
	&\P \Big\{ \min_{\vr \in \mathcal{D}_c^{\circ} \cap S^{m-1}} \max_{\vu \in \mathcal{D}_s \cap S^{n-1}} \left< \mPsi \vu, \vr \right>- \\
	&\hspace{50pt} \E \min_{\vr \in \mathcal{D}_c^{\circ} \cap S^{m-1}} \max_{\vu \in \mathcal{D}_s \cap S^{n-1}} \left< \mPsi \vu, \vr \right> \ge - t \Big\} \\
	&\hspace{50pt}\ge 1- \exp(-t^2/2).
\end{align*}
Putting the above inequality and (\ref{eq:bound for expectation}), (\ref{bound for min min norms}), (\ref{equivalent condition for sufficient condition}), (\ref{equivalence 2}) together, we eventually get that when
$$
\sqrt{m} < \sqrt{\omega^2 \big(\mathcal{D}_s \cap S^{n-1} \big) + \omega^2 \big(\mathcal{D}_c \cap S^{m-1} \big)} - t,
$$
we have
$$
\P \Big\{ \min_{ \vr \in \mathcal{D}_c^{\circ} \cap S^{m-1} } \min_{\vs \in \Ds \times \mathcal{D}_s^{\circ}} \|\vs - \mA^* \vr \|_2 > 0 \Big\} \ge 1- \exp(-t^2/2).
$$
\textbf{Case II: $\vr \notin \Dc^{\circ} \cap S^{m-1}$}. In this case, it is clear that no matter what $\vr$ and $\vs_2$ takes value, it is always holds that
$$
\big\| \vs_2 - \vr \big\|_2^2 > 0.
$$
Thus,
$$
\P \Big\{ \min_{\vr \in S^{m-1} \setminus (\mathcal{D}_c^{\circ} \cap S^{m-1})} \min_{\vs_1 \in \mathcal{D}_s^{\circ}} \|\vs_1 - \mPsi^* \vr \|_2 > 0 \Big\} = 1,
$$
which, by (\ref{equivalent condition for sufficient condition}) and (\ref{equivalence 2}), implies that
$$
\P \Big\{ \min_{ \vr \in S^{m-1} \setminus (\mathcal{D}_c^{\circ} \cap S^{m-1}) } \min_{\vs \in \Ds \times \mathcal{D}_s^{\circ}} \|\vs - \mA^* \vr \|_2 > 0 \Big\} = 1.
$$
\textbf{Union bound}. Combining case I and case II and taking a union bound, we have
$$
\P \Big\{ \min_{ \| \vr \|_2 = 1 } \min_{\vs \in \Ds \times \mathcal{D}_s^{\circ}} \|\vs - \mA^* \vr \|_2 > 0 \Big\} \ge 1- \exp(-t^2/2),
$$
provided
$$
\sqrt{m} < \sqrt{\omega^2 \big(\mathcal{D}_s \cap S^{n-1} \big) + \omega^2 \big(\mathcal{D}_c \cap S^{m-1} \big)} - t.
$$
By Lemma \ref{sufficient condition for failure: step 1} and Lemma \ref{sufficient condition for failure: step 2}, it means that when
$$
\sqrt{m} < \sqrt{\omega^2 \big(\mathcal{D}_s \cap S^{n-1} \big) + \omega^2 \big(\mathcal{D}_c \cap S^{m-1} \big)} - t,
$$
the convex program (\ref{constrained convex program 1}) or (\ref{constrained convex program 2}) fails with probability at least $1-\exp(-t^2/2)$. This completes the proof.

\section{Proof of Lemma \ref{Lipschitz function}} \label{app:B}
To prove Lemma \ref{Lipschitz function}, we only need to show that for any $\mC, \mD \in \R^{m \times n}$
\begin{align*}
	\big| F(\mC) - F(\mD) \big| &= \Big| \min_{\vt \in \Omega_1} \max_{ \vu \in \Omega_2} \left< \mC \vu, \vt \right> - \min_{\vt \in \Omega_1} \max_{ \vu \in \Omega_2} \left< \mD \vu, \vt \right> \Big| \\
	&\le \| \mC - \mD \|_F.
\end{align*}
For any fixed $\vt \in \Omega_1$, let
$$
\vu_0(\vt) \in \argmax_{\vu \in \Omega_2} \left< \mC \vu, \vt \right>.
$$
And we have
$$
\max_{\vu \in \Omega_2} \left< \mD \vu, \vt \right> \ge \left< \mD \vu_0(\vt), \vt \right>.
$$
Then, let
$$
\vt_0 \in \argmin_{\vt \in \Omega_1} \left< \mD \vu_0(\vt), \vt \right>,
$$
and we have
\begin{align*}
  	F(\mC) = \min_{\vt \in \Omega_1} \max_{ \vu \in \Omega_2} \left< \mC \vu, \vt \right> &= \min_{\vt \in \Omega_1} \left< \mC \vu_0(\vt), \vt \right> \\
	&\le \left< \mC \vu_0(\vt_0), \vt_0 \right>.
\end{align*}
Similarly,
\begin{align*}
  	F(\mD) = \min_{\vt \in \Omega_1} \max_{ \vu \in \Omega_2} \left< \mD \vu, \vt \right> &\ge \min_{\vt \in \Omega_1} \left< \mD \vu_0(\vt), \vt \right> \\
	&= \left< \mD \vu_0(\vt_0), \vt_0 \right>.
\end{align*}
Therefore,
\begin{align} \label{difference 1}
	F(\mC) - F(\mD) &\le \left< \mC \vu_0(\vt_0), \vt_0 \right> - \left< \mD \vu_0(\vt_0), \vt_0 \right> \notag\\
	&= \left< (\mC - \mD) \vu_0(\vt_0), \vt_0 \right> \notag\\
	&\le \big\| (\mC - \mD) \vu_0(\vt_0) \big\|_2 \big\|\vt_0 \big\|_2 \notag\\
	&\le \big\| \mC - \mD \big\|_2 \le \| \mC - \mD \|_F.
\end{align}
The same argument gives
\begin{equation} \label{difference 2}
F(\mD) - F(\mC) \le \| \mC - \mD \|_F.
\end{equation}
Thus, combining (\ref{difference 1}) and (\ref{difference 2}), we get
$$
\big| F(\mC) - F(\mD) \big| \le \| \mC - \mD \|_F.
$$
The conclusion follows immediately.

\bibliographystyle{IEEEtran}
\bibliography{IEEEabrv,references}

\end{document}